\documentclass[a4paper,UKenglish,cleveref, autoref,thm-restate]{lipics-v2021}

\usepackage[textsize=tiny]{todonotes}
\newcommand{\citet}[1]{\cite{#1}}
\newcommand{\citep}[1]{\cite{#1}}
\usepackage{xspace}
\usepackage{tikz}
\usepackage{pgfplots}

\newtheorem{assumption}{Assumption}

\newcommand{\field}[1]{\ensuremath{\mathbb{#1}}}
\newcommand{\R}{\ensuremath{\field{R}}} % real numbers
\newcommand{\I}[1]{\ensuremath{\mathbb{I}_{\left\{#1\right\}}}} % indicator function
\newcommand{\tends}{\ensuremath{\rightarrow}} % arrow for limits
\newcommand{\E}{\ensuremath{\mathsf{E}}} % expectation
\newcommand{\defeq}{\ensuremath{\triangleq}}
\newcommand{\Ascr}{\ensuremath{\mathcal A}}

\newcommand{\minimize}{\ensuremath{\mathop{\mathrm{minimize}}\limits}}
\newcommand{\maximize}{\ensuremath{\mathop{\mathrm{maximize}}\limits}}
\usepackage{mathtools}

\provideboolean{submissionversion}
\setboolean{submissionversion}{true}
\provideboolean{blindversion}
\setboolean{blindversion}{true}

\ifthenelse{\boolean{submissionversion}}{
  \renewcommand{\todo}[2][1]{}
  
  \newcommand{\deledit}[1]{}
}{
  
  \newcommand{\deledit}[1]{{\color{orange} \sout{#1}}}
}

\newcommand{\LVR}{\ensuremath{\mathsf{LVR}}\xspace}

\newcommand{\ARB}{\ensuremath{\mathsf{ARB}}\xspace}
\newcommand{\bARB}{\ensuremath{\overline{\mathsf{ARB}}}\xspace}

\newcommand{\VOL}{\ensuremath{\mathsf{VOL}}\xspace}
\newcommand{\bVOL}{\ensuremath{\overline{\mathsf{VOL}}}\xspace}

\newcommand{\FT}{\ensuremath{\mathsf{FT}}\xspace}

\newcommand{\LVF}{\ensuremath{\mathsf{LVF}}\xspace}
\newcommand{\paraheader}[1]{\smallskip\noindent{\sffamily\bfseries #1}}

\title{Loss-Versus-Fair: Efficiency of Dutch Auctions on Blockchains}

\author{Ciamac C.\ Moallemi}{Graduate School of Business, Columbia University, USA \and
  \url{https://moallemi.com/ciamac}}{ciamac@gsb.columbia.edu}
{https://orcid.org/0000-0002-4489-9260}%
{The first author is supported by the Briger Family Digital Finance Lab at Columbia Business
  School, and is an advisor to Paradigm and to fintech companies.}

\author{Dan Robinson}{Paradigm, USA}{dan@paradigm.xyz}{}{}

\authorrunning{C.\,C.\ Moallemi and D.\ Robinson} %TODO mandatory. First: Use abbreviated first/middle names. Second (only in severe cases): Use first author plus 'et al.'

\Copyright{Ciamac C.\ Moallemi and Dan Robinson} %TODO mandatory, please use full first names. LIPIcs license is "CC-BY";  http://creativecommons.org/licenses/by/3.0/


 %TODO mandatory; please add comma-separated list of keywords

\category{} %optional, e.g. invited paper

\relatedversion{} %optional, e.g. full version hosted on arXiv, HAL, or other respository/website
\acknowledgements{The authors thank Eric Budish, Max Resnick, Anthony Zhang for helpful comments.}

\ArticleNo{18}
\begin{document}

\maketitle
%\singlespacing

\begin{abstract}
  Milionis et al.~(2023)
  % \citet{milionis2023automated}
  % cite-style!!!
  studied the rate at which automated market makers leak value to arbitrageurs when block times
  are discrete and follow a Poisson process, and where the risky asset price follows a geometric
  Brownian motion. We extend their model to analyze another popular mechanism in decentralized
  finance for onchain trading: Dutch auctions. We compute the expected losses that a seller incurs
  to arbitrageurs and expected time-to-fill for Dutch auctions as a function of starting price,
  volatility, decay rate, and average interblock time. We also extend the analysis to gradual
  Dutch auctions, a variation on Dutch auctions for selling tokens over time at a continuous
  rate. We use these models to explore the tradeoff between speed of execution and quality of
  execution, which could help inform practitioners in setting parameters for starting price and
  decay rate on Dutch auctions, or help platform designers determine performance parameters like
  block times.
\end{abstract}

%%% Local Variables:
%%% mode: latex
%%% TeX-master: "pricing-dutch-auctions"
%%% End:

\section{Introduction}

Two of the most popular mechanisms for smart contracts to trade tokens are automated market makers
(AMMs) --- in which the price is determined by the contract's reserves --- and Dutch auctions --- in
which the price is determined by the current time.

When block times are discrete, both of these mechanisms leak some value to
arbitrageurs. \citet{milionis2023automated} studied the rate of this value leakage for AMMs, which
is closely related to the concept of ``loss-versus-rebalancing,'' or \LVR \citep{lvr2022}. We
apply a similar analysis to Dutch auctions, deriving a closed form for their ``loss-versus-fair''
(\LVF) --- the expected loss to the seller relative to selling their asset at its contemporaneous
fair price --- as well as their expected time-to-fill. We also do a similar analysis for gradual
Dutch auctions, a variation on Dutch auctions that supports selling tokens at a constant rate over
an extended period of time.

We hope this analysis can help inform practitioners in parameterizing these auctions (e.g.,
choosing the initial price and decay rate) to trade off execution quality with speed of execution,
as well as helping spur research on how to design variants of Dutch auctions that are more resistant to
\LVF.

\paraheader{Dutch auctions.} Also known as descending price auctions, Dutch auctions are auctions
in which an item is listed at a high price that is gradually decreased over time until a bidder
accepts.

Dutch auctions are a commonly-used mechanism in blockchain-based applications, thanks to their
simplicity and efficiency in an environment with high transaction costs, limited privacy, and
pseudonymous identities. Unlike ascending-price or sealed bid auctions, Dutch auctions typically
require only one transaction after they start --- the winning bid, with price as a function of the
block number. This means that failed bidders typically do not need to pay transaction fees such as
``gas'' for their bids or leak any information about their intents.\footnote{One exception is that
  if other bidders attempt to submit a bid at around the same time as the winning bid, their
  transaction may be publicized and/or included on chain after the winning transaction, possibly
  paying fees.} Similarly, the seller needs only consider the single, winning bid, resulting in a
significant reduction in communication and computation complexity versus other auction formats
such as first- or second-price auctions. Dutch auctions are also strongly shill-proof
\citep{komo2024shill}: the seller has no incentive to submit any fake or shill bids to change the
auction outcome.

For these reasons, Dutch auctions have been used for a variety of applications in decentralized
finance:
\begin{itemize}
\item Liquidations in peer-to-pool lending protocols like Maker \citep{makerdaoliquidation}
  or
  Ajna \citep{ajnawhitepaper}
\item Rolling loans and discovering interest rates in peer-to-peer lending protocols like Blend
  \citep{blend}
\item Routing trades in request-for-quote (RFQ) protocols like
  UniswapX \citep{adams2023uniswapx}
  and 1inch Fusion \citep{1inchfusion}
\item Collecting and converting fees in protocols like Euler \citep{eulerv2}
\end{itemize}

Dutch auctions can be used both for price discovery of highly illiquid assets like NFTs, and for
automated execution between liquid assets. Here, we focus on the latter case, and specifically on
auctions between highly liquid but volatile pairs of tokens, such as between ETH and stablecoins
like USDC. In particular, we assume a \emph{common value} setting where all potential buyers agree
on the value of the asset being sold at any point in time (because, for example, the asset may be
liquidly traded in other off-chain markets), and assume the price of the asset obeys \emph{geometric Brownian motion}.

\paraheader{Gradual Dutch auctions.}  Gradual Dutch auctions (GDAs) are a variation on Dutch
auctions that were introduced by \citet{gda2022} as a mechanism for selling NFTs or tokens at a
constant target rate over an extended time period. ``Continuous gradual Dutch auctions'' (CGDAs),
the kind we consider in this paper, could be thought of as a series of infinitesimal Dutch
auctions of a fungible token, with new auctions being initiated at a linear rate over time, and
each auction independently decaying in price at an exponential rate.

\paraheader{Arbitrage profits.}  Dutch auctions have some drawbacks when implemented on a
blockchain. In particular, since blocks only arrive at discrete times, the true market price of
the asset at the time a block is created may be higher than the price offered by the auction, due
to the decay of the Dutch auction price and the drift and volatility of the asset. This means the
seller should expect to sell the asset at a discount to the market price or fair value at the time
of sale, with the profits going to arbitrageurs or whoever is able to capture value from ordering
transactions in the block --- a type of maximal extractable value (MEV).

This type of loss is similar to the ``quote-sniping'' losses of market makers in high-frequency
trading models \citep{budish2015high}, or the ``loss-versus-rebalancing''  suffered by liquidity
providers on automated market makers, a concept introduced by \citet{lvr2022}. In
\citet{milionis2023automated}, \LVR was extended to incorporate discrete blocks. For
analytic tractability, block generation times are assumed to be from a Poisson process, an
assumption we also make here.

\paraheader{Contributions.}
In this paper, we apply a similar model to Dutch auctions and  gradual Dutch
auctions. Given certain parameters for a Dutch auction --- volatility, drift, starting price, decay
rate, and average block arrival times --- we derive closed-form expressions for both the losses to
fair value and expected time-to-fill. We also extend the analysis to gradual Dutch auctions,
showing how expected losses to arbitrageurs and expected sales rate vary as a function of these
parameters.

For both Dutch auctions and gradual Dutch auctions, as long as the auction starts above the
current price, \LVF is given by the following expression (where $\delta$ is the decay rate of the auction plus the asset's drift in log space, $\sigma$ is the volatility of the asset, and $\Delta t$ is the mean interblock time):

\[
    \LVF_+ = \frac{1}{1+\frac{\delta}{\sigma^2} \left( \sqrt{1 + \frac{2 \sigma^2}{\delta^2 \Delta t} }
      - 1\right)}.
\]
For example, if volatility is $5\%$ per day ($0.017\%$ per second), decay rate is $0.01\%$ per second, and average block time is 12 seconds, $\LVF_+$ is about $0.13\%$. This would mean that for every $\$100$ worth of tokens that they sell, the seller should expect to get about $\$99.87$.

\todo{add more stuff, asymptotic behavior}

For regular Dutch auctions, the amount of time to fill if the starting price of the
auction is higher than the current price is given by the following formula, in which $z_0$ is the
(log) difference between the starting price and the current price:
\[
  \FT(z_0) = \frac{z_0}{\delta} +
  \frac{\Delta t}{2} \left( 1 + \sqrt{1 + \frac{2 \sigma^2}{\delta^2 \Delta t}} \right).
\]
For example, with the same parameters as above, and with starting price $0.1\%$ higher than the
current price, the expected time to fill is about 23.3 seconds.  For gradual Dutch auctions, we
find a closed form expression for the rate at which the asset is sold over time.

We also find closed forms for $\LVF$ and $\FT$ in the cases where the starting price of the auction is below the current price.

These models show how changing the decay rate of the auction affects both speed of execution and
expected loss, helping inform practitioners who want to trade off between those values when
choosing auction parameters such as initial price and decay rate. They also show how the
characteristics of the blockchain --- particularly average block time --- affect the efficiency of
Dutch auctions. For example, the formula for $\LVR_+$ above satisfies the lower bound
\[
  \LVF_+ \geq \frac{1}{1 + \frac{1}{\sigma \sqrt{\Delta t / 2}}}
  \approx \sigma \sqrt{\Delta t / 2},
\]
where the approximation holds for $\Delta t$ small (the ``fast block'' regime).
This suggests that if a platform wants to
support Dutch auctions that lose less than 2 basis points for assets with daily volatility of
5\%, it will need to have block times of less than 2.75 seconds.

\iffalse
\paraheader{Contributions.} The contributions of this paper are as follows:
\todo{add contributions}

\begin{enumerate}
\item quantify loss and time-to-fill for dutch auctions
\item loss can be decomposed into: piece that's due to volatility, and then additional loss from
  drift
\item can be used to set parameters: if price accurately know, initial should be the same as
  initial fair value, drift should be set to balance time to fill and loss
\item quantify loss and trading in cts dutch auctions
\end{enumerate}
\fi

\subsection{Literature Review}

Dutch auctions have been analyzed extensively in the auction theory and mechanism design
literature, since at least the work of \citet{vickrey1961counterspeculation}, who showed the
strategic equivalence of Dutch auctions and first-price sealed-bid auctions under certain
assumptions.

Our approach is related to barrier-diffusion approaches \citep{hasbrouck2007empirical} to limit
order pricing. For example, in \citet{LO200231}, the time-to-fill for a limit order is modelled as
the first-passage time for a geometric Brownian motion with drift, and solve for the distribution
of this time. Mathematically, this is equivalent to a continuous time version of our
model.\footnote{In our setting, the drift arises from descending price of the Dutch auction, while
  for \citet{LO200231}, the limit order is at a static price and the drift arises from the
  underlying asset price process.} Crucially, they do not consider loss-versus-fair for a limit
order, since this quantity would be zero in continuous time. \citet{latency2009} consider
frictions introduced by latency in submitting limit orders, also using a barrier-diffusion
model. In that setting, latency acts as a friction that limits the ability of an agent to update
their limit orders in a timely fashion, in reaction to changing market conditions. The central novelty of the present paper is the blockchain setting: we analyze frictions
restricting the ability to trade in the auction introduced by the discrete block generation
process. To our knowledge, no prior work has modeled the behavior of Dutch auctions for geometric
Brownian motion assets with discrete block generation times.

The idea of gradual Dutch auctions was proposed by \citet{gda2022}. \citet{vrgda2022} proposed an
extension on the idea, variable rate GDAs, in which the target sales rate could vary as a function
of time. Gradual Dutch auctions could be thought of as similar to automated market makers (AMMs)
for which the price impact function is an exponential function, the fee to buy is $0$,
the fee to sell is infinite, and the asset price has a negative drift. In this way, we build on
the setting of \citet{milionis2023automated} in computing arbitrage profits for AMMs with fees.

% Milionis et
% al. [2023] briefly considered the case of arbitrage profits on AMMs for assets with non-zero
% drift, in Appendix C. If you set γ- to 0 and y+ to infinity, their stationary distribution and
% probability of trading matches the one we find for Gradual Dutch Auctions. They do not compute an
% expression for the expected loss of an AMM in that case (much less one with an exponential price
% impact function), nor do they find any expressions for expected time to trade.

A version of the GDA mechanism was studied by \citet{kulkarni2023credibility}. That work considers
the use of discrete GDAs for illiquid NFTs where buyers depend on private signals for valuation,
rather than continuous GDAs for highly liquid fungible tokens driven by common valuations.

%%% Local Variables:
%%% mode: latex
%%% TeX-master: "pricing-dutch-auctions"
%%% End:

\section{Model}

We imagine a scenario where an agent is selling\footnote{While we focus on the case of an agent
  selling the asset via a Dutch auction, our model also applies to the case of an agent buying via
  an ascending price Dutch auction-style mechanism. In that case, the mechanism would have a
  steadily increasing bid price at which it is willing to buy the asset, and analogous formulas
  could be obtained. Note that over longer time horizons, the two cases are not completely
  symmetric because of the positivity of prices and the inherent asymmetry of geometric Brownian
  model. In particular, for example, for a seller \LVF is bounded above by 100\% since the sale
  price will always be bounded below by zero, while \LVF is unbounded above for a buyer, since the
  buy price is unbounded above.}  a risky asset via a descending price Dutch auction in a common
value setting. Following the model of \citet{milionis2023automated}, we assume there exists the
common fundamental value or price $P_t$ at time $t$ that follows a geometric Brownian motion price
process,
\begin{equation}\label{eq:gbm}
  \frac{dP_t}{P_t} = \mu\, dt + \sigma\, dW_t,
\end{equation}
where $\{W_t\}$ is a Brownian motion, and the process is parameterized by drift $\mu$ and
volatility $\sigma>0$.

The agent is progressively willing to lower their offered price. Let $A_t$ denote the lowest price the
agent is willing to sell at, at time $t$, i.e., the best ask price. We assume $A_t$ decreases
exponentially according to\footnote{The spirit here is to model a Dutch auction where the price is
  decreasing at a constant rate. We specifically choose exponentially decreasing prices (i.e.,
  prices that are decreasing at a constant relative rate) because it matches well with the
  geometric Brownian motion price dynamics \eqref{eq:gbm}. An alternative choice would be to
  assume the ask price decreases linearly and that the price process is a arithmetic Brownian
  motion. On the short timescales of practical interest, this would yield similar results both
  quantitatively and qualitatively to the model here. To see this, note that, under our model,
  $A_t = A_0 e^{- \lambda t} \approx A_0 \left( 1 - \lambda t \right)$, for $t$ small.}
\begin{equation}\label{eq:dAdt}
\frac{dA_t}{A_t} = - \lambda \, dt,
\end{equation}
with decay constant $\lambda > 0$. Define the log mispricing process $z_t \defeq \log(A_t/P_t)$,
so that, applying It\^o's lemma,
\[
  dz_t = -\underbrace{\left(\lambda + \mu - \tfrac{1}{2} \sigma^2 \right)}_{\defeq \delta}\, dt
  + \sigma \, dW_t.
\]

We assume that blocks are generated according to a Poisson process\footnote{For a proof-of-work
blockchain, Poisson block generation is a natural assumption \citep{nakamoto2008bitcoin}. However,
modern proof-of-stake blockchains typically generate blocks at deterministic times. In these
cases, we will view the Poisson assumption as an approximation that is necessary for tractability.}
of rate
$\Delta t^{-1}$, where $\Delta t > 0$ is the mean interblock time. We assume there is a population
of ``arbitrageurs'', or traders informed about the fundamental price $P_t$, who will buy from the
agent at any discount to this price. However, these agents can only act at block generation times.

Thus, if $\tau$ is a block generation time, and\footnote{We assume that the processes
  $\{A_t,P_t\}$ are right continuous with left limits, and define $A_{\tau-}$ and $P_{\tau-}$ to
  be the left limits, i.e., the values immediately before the time $\tau$.}
$A_{\tau-} < P_{\tau-}$, arbs trade until there is
no marginal profit, and so that the ask price updates with $A_\tau = P_\tau$ and $z_\tau = 0$. Thus, we have
$z_{\tau} = \max( 0, z_{\tau_-} )$.  Then, $\{ z_t \}$ is a Markov jump diffusion
process. Since it involves the interaction of a diffusive process $\{z_t\}$ with a barrier
($z_t=0$), our model falls into the general class of barrier-diffision models for market
microstructure \citep{hasbrouck2007empirical}.

This model is grounded in the typical way that blocks are built on decentralized blockchains
today, in which each block is built by a ``proposer'' or ``miner''. We imagine that the Dutch
auction is implemented via a smart contract that sets the minimum acceptable price as a function
of either block time or the block height. Within a block, the first transaction willing to
pay the price will succeed. In the case that there are multiple buyers (as might be expected if
the publicly observable fair value exceeds the limit price of the auction at the time), they
would compete for earlier block position by offering priority fees to the proposer. We assume
that each block proposer is independent and short-term profit-maximizing,\footnote{Note that if
  this assumption is violated --- such as if a single proposer controls multiple blocks in a row
  --- they may be able to extract additional profit at the expense of the seller. As of this
  writing, extraction of this kind of ``multi-block MEV'' is generally believed to be rare on
  major blockchains like Ethereum, although there are reasons to be concerned that it could
  increase in the future.} and hence they would include and prioritize the top priority-fee-paying
transaction in the block, allowing that buyer to win the trade. In this case some or all of the
arbitrage profits may accrue not to the buyer, but instead to the proposer in the form of
priority fees. Our focus in this paper, however, is on quantifying the loss to the seller and not
how it is distributed.

We will make the following assumption:
\begin{assumption}\label{as:stationary}
  Assume that $\delta \defeq \lambda + \mu - \tfrac{1}{2} \sigma^2 > 0$.
\end{assumption}
This assumption is sufficient to ensure that trade occurs with probability $1$, and necessary so
that the expected time to trade is finite. It can be satisfied by the agent making a
sufficiently large choice of the decay rate $\lambda$.  Under \Cref{as:stationary}, the following
lemma gives the stationary distribution $\pi(z)$ of $z_t$:\footnote{While applied in a different
  context, \Cref{lem:stationary} is a special case of Theorem~7 of \citet{milionis2023automated} up
  to a sign change, with $\gamma_-\tends \infty$. For completeness, a standalone proof is provided
  in \Cref{sec:proofs}.}
\begin{lemma}\label{lem:stationary}
  If $\delta > 0$, the process $z_t$ is an ergodic process on $\R$, with unique invariant distribution
  $\pi(\cdot)$ given by the density
  \[
    p_\pi(z) =
    \begin{cases}
      \pi_+ \times p^{\exp}_{\zeta_+}(z) & \text{if $z \geq 0$}, \\
      \pi_- \times p^{\exp}_{\zeta_-}(-z) & \text{if $z < 0$},
    \end{cases}
  \]
  for $z \in \R$. Here, $p^{\exp}_{\zeta}(z) \defeq \zeta e^{-\zeta z}$ is the density of an
  exponential distribution over $z\in\R_+$ with parameter $\zeta>0$. The parameters
  $\{\zeta_\pm\}$ are given by
  \[
    \zeta_- \defeq \frac{\delta}{\sigma^2} \left( \sqrt{1 + \frac{2 \sigma^2}{\delta^2 \Delta t} }
      - 1\right),\quad
    \zeta_+ \defeq \frac{2  \delta}{ \sigma^2 }.
  \]
  The probabilities $\{ \pi_\pm\}$ are given by
  \[
  \pi_-  \defeq \pi\big((-\infty,0)\big)  = \delta \Delta t \zeta_-,
  \quad
  \pi_+  \defeq \pi\big([0,+\infty)\big)  = 1 - \delta \Delta t \zeta_-.
\]
\end{lemma}

%%% Local Variables:
%%% mode: latex
%%% TeX-master: "pricing-dutch-auctions"
%%% End:

\section{Regular Dutch Auctions}\label{sec:rda}

We first consider the case of a discrete quantity of the risk asset for sale, initially at ask
price $A_0$, or, alternatively, initial log mispricing $z_0 \defeq \log(A_0/P_0)$, with the ask
price $A_t$ decreasing exponentially at rate $\lambda$ according to \eqref{eq:dAdt}. Suppose the
order is traded at fill time $\tau_F$, i.e., $\tau_F$ is the earliest block generation time which
satisfies $z_{\tau_F}\leq 0$. Then, the order will sell at price $A_{\tau_F}$
when the fundamental value is $P_{\tau_F}$. We are interested in the expected relative loss versus the
fundamental price or fair value, i.e.,
\[
\frac{P_{\tau_F} - A_{\tau_F}}{P_{\tau_F}} = 1 - e^{z_{\tau_F}}.
\]

\paraheader{Loss-versus-fair and time-to-fill.}
We are interested in the expected relative loss, which we call ``loss-versus-fair'' (\LVF), i.e.,
\[
\LVF(z_0) \defeq \E\left[ \left. 1 - e^{z_{\tau_F}} \right| z_0\right].
\]
We are also interested in the expected time-to-fill, i.e.,
\[
\FT(z_0) \defeq \E\left[ \left. \tau_F \right| z_0  \right].
\]
The following theorem characterizes these quantities:
\begin{theorem}\label{th:regular}
  If $z_0 \geq 0$, the expected relative loss and time-to-fill are given by
  \begin{equation}\label{eq:lvf-p}
    \LVF(z_0) = \frac{1}{1+\zeta_{-}}
    =
    \frac{1}{1+\frac{\delta}{\sigma^2} \left( \sqrt{1 + \frac{2 \sigma^2}{\delta^2 \Delta t} }
      - 1\right)}
    \defeq \LVF_+,
  \end{equation}
  \begin{equation}\label{eq:ft-p}
    \FT(z_0) = \frac{z_0}{\delta} +
    \frac{\Delta t}{2} \left( 1 + \sqrt{1 + \frac{2 \sigma^2}{\delta^2 \Delta t}} \right).
  \end{equation}
  If $z_0 < 0$, then
  \begin{equation}\label{eq:lvf-n}
    \LVF(z_0)
    =
      1 -
      \frac{e^{z_0}}{1 + \Delta t \left( \delta - \tfrac{1}{2} \sigma^2  \right) }
      + \left(
        \frac{1}{1+\zeta_{-}}
        -
        \frac{\Delta t \left( \delta - \tfrac{1}{2} \sigma^2 \right)}
        {1 + \Delta t \left(\delta -  \tfrac{1}{2}
          \sigma^2 \right)}
    \right)
    e^{\zeta_{-} z_0},
  \end{equation}
  \begin{equation}\label{eq:ft-n}
    \FT(z_0) = \frac{\Delta t}{2} \left( 2 +
      \left(  \sqrt{1 + \frac{2 \sigma^2}{\delta^2 \Delta t}} - 1 \right) e^{\zeta_{-} z_0} \right).
  \end{equation}
\end{theorem}

The formulas of \Cref{th:regular} are illustrated for representative parameter choices in
\Cref{fig:lvf-ft}.

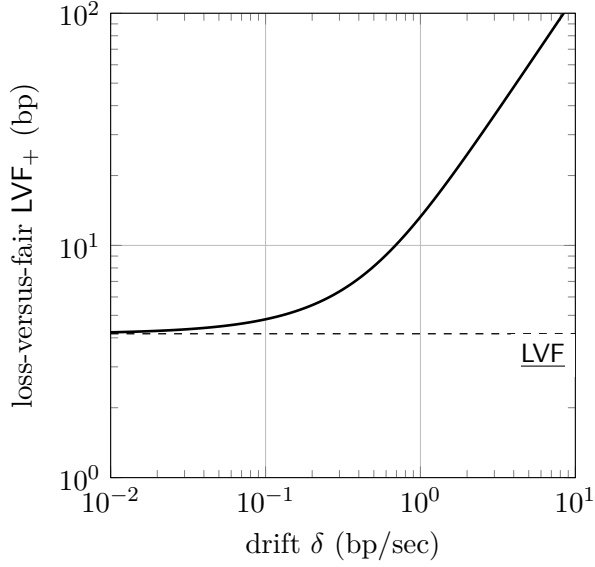
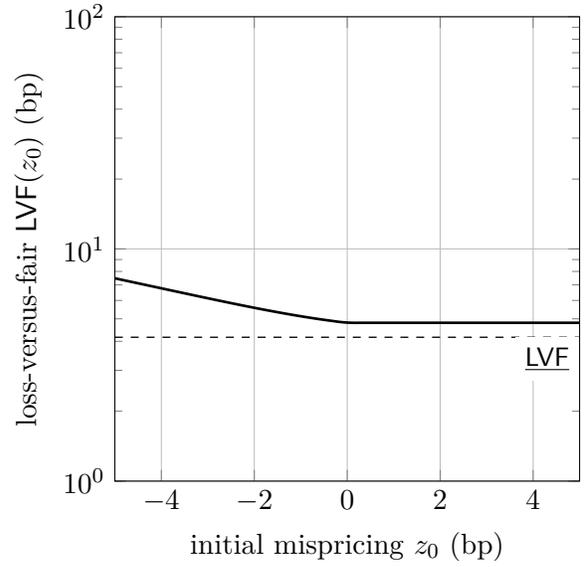
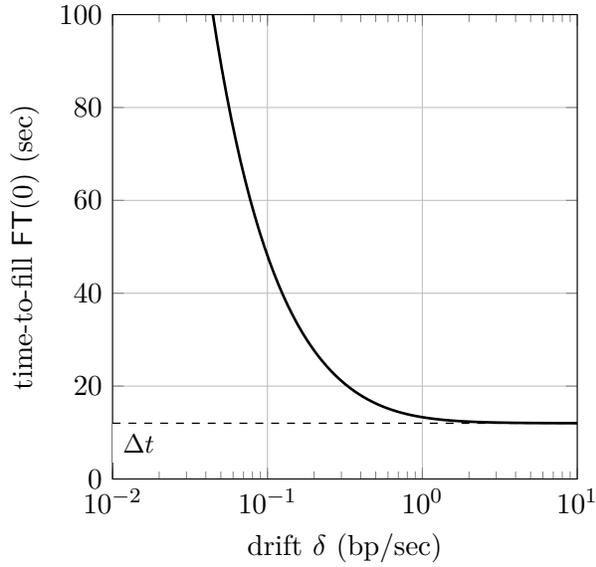
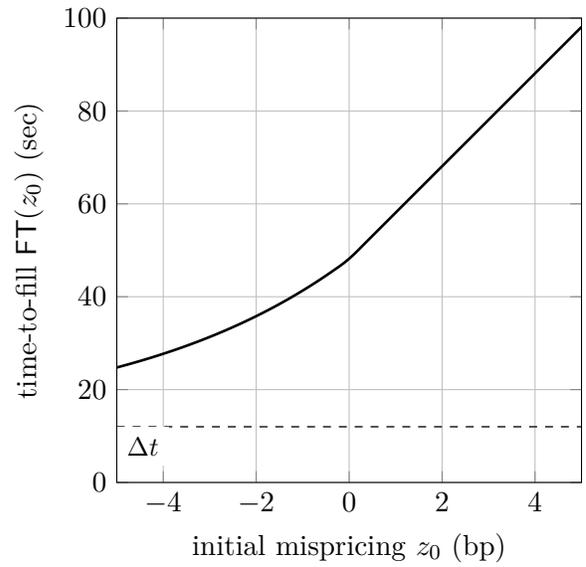
\begin{figure}
  \centering
  \begin{subfigure}[t]{0.47\textwidth}
    \begin{tikzpicture}
      \begin{axis}[
        xlabel={drift $\delta$ (bp/sec)},
        ylabel={loss-versus-fair $\LVF_+$ (bp)},
        xmode=log,
        ymode=log,
        ymin=1,
        ymax=100,
        xmin=0.01,
        xmax=10,
        width=\textwidth,
        height=\textwidth,
        grid=major,
        ]

        \addplot
        [mark=none,line width=1pt,smooth]
        table
        [x expr=\thisrow{delta}*1e4,y expr=\thisrow{lvf}*1e4, col sep=comma,mark=none]
        {figures/vary-delta.csv};

        %\addplot[mark=none,dashed,line width=0.5pt,domain=0.01:10] {4.1649312786339027}
        %node [pos=1,above left,rectangle,fill=white] {\small $\LVF(z_0) \geq \underline{\LVF}$};

        \addplot[mark=none,dashed,line width=0.5pt,domain=0.01:10] {4.1649312786339027}
        node [pos=1,below left,rectangle,fill=white] {\small $\underline{\LVF}$};
      \end{axis}
    \end{tikzpicture}
    \caption{$\LVF_+=\LVF(0)$ as a function of $\delta$.
      \label{fig:lvf-delta}}
  \end{subfigure}
  \qquad
  \begin{subfigure}[t]{0.47\textwidth}
    \begin{tikzpicture}
      \begin{axis}[
        xlabel={initial mispricing $z_0$ (bp)},
        ylabel={loss-versus-fair $\LVF(z_0)$ (bp)},
        % xmode=log,
        ymode=log,
        ymin=1,
        ymax=100,
        xmin=-5,
        xmax=5,
        width=\textwidth,
        height=\textwidth,
        grid=major,
        ]

        \addplot
        [mark=none,line width=1pt,smooth]
        table
        [x expr=\thisrow{z0}*1e4,y expr=\thisrow{lvf}*1e4, col sep=comma,mark=none]
        {figures/vary-z0.csv};

        \addplot[mark=none,dashed,line width=0.5pt,domain=-5:5] {4.1649312786339027}
        node [pos=1,below left,rectangle,fill=white] {\small $\underline{\LVF}$};
      \end{axis}
    \end{tikzpicture}
    \caption{$\LVF(z_0)$ as a function of $z_0$, assuming a fixed value of
      $\delta=0.1\text{ (bp/sec)}$.
      \label{fig:lvf-v0}}
  \end{subfigure}
  \\
  \begin{subfigure}[t]{0.47\textwidth}
    \begin{tikzpicture}
      \begin{axis}[
        xlabel={drift $\delta$ (bp/sec)},
        ylabel={time-to-fill $\FT(0)$ (sec)},
        xmode=log,
        %ymode=log,
        ymin=0,
        ymax=100,
        xmin=0.01,
        xmax=10,
        width=\textwidth,
        height=\textwidth,
        grid=major,
        ]

        \addplot
        [mark=none,line width=1pt,smooth]
        table
        [x expr=\thisrow{delta}*1e4,y expr=\thisrow{ft}, col sep=comma,mark=none]
        {figures/vary-delta.csv};

        \addplot[mark=none,dashed,line width=0.5pt,domain=0.01:10] {12}
        node [pos=0,below right,rectangle,fill=white] {\small $\Delta t$};
      \end{axis}
    \end{tikzpicture}
    \caption{$\FT(0)$ as a function of $\delta$.
      \label{fig:ft-delta}}
  \end{subfigure}
  \qquad
  \begin{subfigure}[t]{0.47\textwidth}
    \begin{tikzpicture}
      \begin{axis}[
        xlabel={initial mispricing $z_0$ (bp)},
        ylabel={time-to-fill $\FT(z_0)$ (sec)},
        % xmode=log,
        % ymode=log,
        ymin=0,
        ymax=100,
        xmin=-5,
        xmax=5,
        width=\textwidth,
        height=\textwidth,
        grid=major,
        ]

        \addplot
        [mark=none,line width=1pt,smooth]
        table
        [x expr=\thisrow{z0}*1e4,y expr=\thisrow{ft}, col sep=comma,mark=none]
        {figures/vary-z0.csv};

        \addplot[mark=none,dashed,line width=0.5pt,domain=-5:5] {12}
        node [pos=0,below right,rectangle,fill=white] {\small $\Delta t$};
      \end{axis}
    \end{tikzpicture}
    \caption{$\FT(z_0)$ as a function of $z_0$, assuming a fixed value of
      $\delta=0.1\text{ (bp/sec)}$.
      \label{fig:ft-v0}}
  \end{subfigure}

  \caption{Comparison of \LVF and \FT for different parameter choices. These figures assume
    $\sigma=5\%\text{ (daily)}$ and $\Delta t = 12\text{ (sec)}$.
    The dashed lines correspond to the lower bounds of \eqref{eq:lvf-bound} and
    \eqref{eq:ft-bound}.
    \label{fig:lvf-ft}}
\end{figure}

\paraheader{Discussion of loss-versus-fair.}  Observe that, for $z_0 \geq 0$, the loss is given by
$\LVF(z_0)=\LVF_+$ and does not depend on the initial mispricing $z_0$. This is because, starting
at $z_0 \geq 0$, the mispricing process must first pass through the boundary $z_t=0$, since it is
continuous. If we denote by $\tau_0$ the first passage time of that boundary, because $\{z_t\}$ is
a Markov process and block generation times are memoryless, we have that
$\LVF(z_0) = \E[\LVF(z_{\tau_0})] = \LVF(0) = \LVF_+$. For $z_0 < 0$, $\LVF(z_0)$ is strictly
decreasing in $z_0$. This is intuitive: the more the asset is initially underpriced, the larger
the expected losses experienced upon the eventual sale.

Now, consider properties of the loss $\LVF_+$. Observe that this is a strictly increasing function
of the mispricing $\delta$, so that it is minimized when $\delta=0$, and we have the lower bound
\begin{equation}\label{eq:lvf-bound}
  \underline{\LVF} \defeq \frac{1}{1 + \frac{1}{\sigma \sqrt{\Delta t/2}}} \leq \LVF_+
  \leq \LVF(z_0).
\end{equation}
In general, setting as small a value of the drift $\delta$ as possible minimizes losses. However,
the left side of \eqref{eq:lvf-bound} yields a lower bound on the loss that is due intrinsic
volatility and discrete blocks. Indeed, in the fast block regime, when the average interblock time
$\Delta t$ is small, this lower bound takes the form
\[
  \underline{\LVF} \defeq \frac{1}{1 + \frac{1}{\sigma \sqrt{\Delta t/2}}} \approx \sigma \sqrt{\Delta t/2},
\]
which is the standard deviation of changes in the mispricing process over half of a typical
interblock time. This is a minimum, unavoidable level of loss, no matter what choice of auction
parameters $(z_0,\delta)$ is made.

We can also consider the behavior of $\LVF_+$ as a function of the volatility $\sigma$. It is
straightforward to see that $\LVF_+$ is increasing in $\sigma$, and hence is lower bounded by the
value as $\sigma\tends 0$, i.e.,
\begin{equation}\label{eq:lvf-bound2}
\LVF_+ \geq \lim_{\sigma \tends 0} \LVF_+ = \frac{1}{1 + \frac{1}{\delta \Delta t}} \approx \delta
\Delta t,
\end{equation}
where the final approximation holds in the fast block regime when $\Delta t$ is small. The lower
bound on the right side of \eqref{eq:lvf-bound2} is the price decay over a
single block. In the fast block regime, this is a lower bound on $\LVF_+$, which also includes the
impact of volatility.

\paraheader{Discussion of time-to-fill.}  For the time-to-fill, observe that $\FT(z_0)$ is a
strictly increasing function of the initial mispricing $z_0$ and a strictly decreasing function of
the drift $\delta$, and that
\begin{equation}\label{eq:ft-bound}
  \FT(z_0) \geq \lim_{z\tends -\infty} \FT(z) = \Delta t.
\end{equation}
This lower bound is intuitive: by the memoryless nature of the Poisson process, the time-to-fill is
always lower bounded by the mean interblock time.

$\FT(z_0)$ is also increasing as a function of the volatility $\sigma$, hence we have the lower
bound
\[
  \FT(z_0) \geq \lim_{\sigma \tends 0} \FT(z_0) = \frac{z_0}{\delta} + \Delta t.
\]
This bound is also intuitive: absent volatility, the ask price must first drift to the fair price
(in time $z_0/\delta$), and then wait for the next block (in time $\Delta t$).

\paraheader{Parameter optimization (known value).} \Cref{th:regular} can be applied to optimize
the initial auction price $A_0$ at time $t=0$ and the decay rate $\lambda$. When the initial value
$P_0$ is known, we will parameterize this decision with the variables
$z_0 \defeq \log(A_0/P_0)$ and $\delta \defeq \lambda + \mu - \tfrac{1}{2} \sigma^2 >
0$. Then, the seller can solve the optimization problem
\[
  \minimize_{z_0,\delta \geq 0}\ \LVF(z_0) + \theta \cdot \FT(z_0).
\]
Here, $\theta \geq 0$ is a parameter that captures the trade off between minimizing loss and
time-to-fill. The efficient frontier of Pareto optimal outcomes with these two objectives can
be generated by varying $\theta$. An example of such an efficient frontier is illustrated in \Cref{fig:eff}.

Note that, in this setting, it is never optimal to pick $z_0 > 0$. This is because such a choice of
$z_0$ is Pareto dominated by setting $z_0 = 0$: in this case lowering the value of $z_0$ strictly
decreases $\FT(z_0)$, without increasing $\LVF(z_0)$. Indeed, with the representative parameter
choices of \Cref{fig:eff}, setting $z_0 \approx 0$ is typically optimal, i.e., the auction should
be started at the current fundamental value (when it is known).

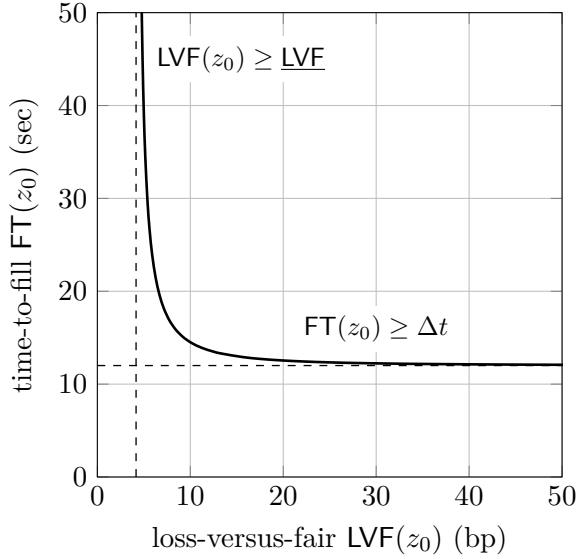
\begin{figure}
  \centering
  \begin{tikzpicture}
    \begin{axis}[
      xlabel={loss-versus-fair $\LVF(z_0)$ (bp)},
      ylabel={time-to-fill $\FT(z_0)$ (sec)},
      % xmode=log,
      % ymode=log,
      ymin=0,
      ymax=50,
      xmin=0,
      xmax=50,
      width=0.47\textwidth,
      height=0.47\textwidth,
      grid=both,
      ]

      % \addplot
      % [mark=none,line width=1pt,smooth]
      % table
      % [x expr=\thisrow{lvf}*1e4,y expr=\thisrow{ft}, col sep=comma,mark=none]
      % {figures/efficient-frontier-sigma-0.000500.csv};

      % \addplot
      % [mark=none,line width=1pt,smooth]
      % table
      % [x expr=\thisrow{lvf}*1e4,y expr=\thisrow{ft}, col sep=comma,mark=none]
      % {figures/efficient-frontier-sigma-0.005000.csv};

      \addplot
      [mark=none,line width=1pt,smooth]
      table
      [x expr=\thisrow{lvf}*1e4,y expr=\thisrow{ft}, col sep=comma,mark=none]
      {figures/efficient-frontier.csv};

      \addplot[mark=none,dashed,line width=0.5pt,domain=0:50] {12}
      node [pos=0.6,above=5pt,rectangle,fill=white] {\small $\FT(z_0) \geq \Delta t$};
      \addplot[mark=none,dashed,line width=0.5pt]
      coordinates {(4.1649312786339027,0) (4.1649312786339027,50)}
      node [pos=0.9,right=4pt,rectangle,fill=white] {\small $\LVF(z_0) \geq \underline{\LVF}$};
    \end{axis}
  \end{tikzpicture}
  \caption{The efficient frontier trading off loss-versus-fair and time-to-fill.
This figure assumes
    $\sigma=5\%\text{ (daily)}$ and $\Delta t = 12\text{ (sec)}$.
    The dashed lines correspond to the lower bounds of \eqref{eq:lvf-bound} and \eqref{eq:ft-bound}.
    \label{fig:eff}}
\end{figure}

\paraheader{Parameter optimization (unknown value).} Another setting of interest is where the
buyer is uncertain of the value $P_0$ when determining the auction parameters. We describe this
uncertainty with a lognormal Bayesian prior: assume the seller believes that
$P_0 \sim \hat P_0 e^{-\tfrac{1}{2} \sigma_0^2 + \sigma_0 Z}$, where $Z \sim N(0,1)$,
$\hat P_0 = \E[P_0]$ is the mean of the prior belief, and $\sigma_0 > 0$ is the volatility of the
prior belief. Then, we have $z_0 = \log A_0/\hat P_0 + \tfrac{1}{2} \sigma_0^2 - \sigma_0 Z$.
Then, the seller can compute the loss-versus-fair and time-to-fill efficient frontier by solving the
optimization problem
\[
  \minimize_{A_0,\delta \geq 0}\ \E\left[ \LVF\left(\log A_0/\hat P_0 + \tfrac{1}{2} \sigma_0^2 - \sigma_0
    Z\right) \right]
  + \theta \cdot \E\left[ \FT\left(\log A_0/\hat P_0 + \tfrac{1}{2} \sigma_0^2 - \sigma_0 Z\right) \right],
\]
for varying values of $\theta \geq 0$. Note that the expectations in the objective function can be
computed in closed form, these formulas are provided in \Cref{sec:parameter}.

\iffalse
TODO:
\begin{itemize}
\item
TODO:

return over half a block length, cts time all is zero

  lower bound is losses from volatility and no drift

  second term adds drift correction

\item look at $\sigma\tends 0$, $\delta > 0$ (drift over a block)

\item fast block regime, does not depend on $\delta$
\end{itemize}
\fi

%%% Local Variables:
%%% mode: latex
%%% TeX-master: "pricing-dutch-auctions"
%%% End:

\section{Gradual Dutch Auctions}\label{sec:gda}

In this section, we develop stationary, steady-state analogs of the results of \Cref{sec:rda} in
the context of gradual Dutch auctions. Introduced by \citet{gda2022}, the continuous gradual Dutch
auctions we consider here continously emit the risky asset for sale at a rate per unit time given by
$r > 0$. Each emission is in turn are sold through a Dutch auction where the price decreases
exponentially with decay rate $\lambda > 0$. Our goal will be to compute the steady-state rate at
which such auctions leak value to arbitrageurs, as well as the rate of trade. We will see a
similar tradeoff as in \Cref{sec:rda}.

In our stationary, steady-state setting, we will imagine that the seller has been continuously
emitting auctions since time $t=-\infty$. At any time $t$, if an auction has age $u$, the auction
price is given by $k e^{-\lambda u}$, for some constant $k > 0$. When the the age of the oldest
available auction is $T$, this auction defines the best ask price by $A_t=k e^{-\lambda T}$.
Hence, if an agent wishes to purchase a total quantity $q$ at time $t$, and the age of the oldest available
auction is $T$, the total cost is given by
\[
C_t(q) = \int_{T-q/r}^T ke^{-\lambda u} \cdot r\, dt = \frac{k r}{\lambda} \frac{e^{\lambda q/r} -
  1}{e^{\lambda T}} = A_t \cdot \frac{r}{\lambda} \left( e^{\lambda q/r} - 1 \right).
\]

Denote the block generation times by $0 < \tau_1 < \tau_2 < \cdots$. When a block is generated at
each time $t=\tau_i$, arbitrageurs can trade against the auctions, and will myopically seek to do
so to maximize their instananeous profit, assuming they value the risky asset at the current
fundamental price $P_t$. The following lemma characterizes this behavior:

\begin{lemma}\label{lem:myopic}
  Suppose a block is generated at time $\tau$, with current fundamental price given by $P \defeq P_\tau$, and
  mispricing (immediately before block generation) given by  $z \defeq z_{\tau-}$. Then, if
  $\lambda > 0$, the
  optimal arbitrage trade quantity of the risky asset is given by
  \[
    q^*(z) \defeq - \frac{r}{\lambda} z \I{z \leq 0},
  \]
  with optimal arbitrage profits (or, equivalently, the total loss experienced by the auction
  seller relative to selling at the current fair fundamental price $P$)
  \[
    A^*(P,z)  \defeq \frac{Pr}{\lambda} \left\{ e^z - 1 - z \right\}\I{z \leq 0}.
  \]
\end{lemma}
\begin{proof}
  The arbitrageur faces the maximization problem
  \[
    \begin{array}{ll}
      \maximize_{q \geq 0} & \displaystyle P_\tau q - C_{\tau-}(q) =  P \left\{  q - \frac{e^z r}{\lambda} \left( e^{\lambda
                             q / r} - 1 \right)  \right\},
    \end{array}
  \]
  where we use the fact that $A_{\tau-} = P_\tau e^{z_{\tau-}}$.  The result follows from
  straightforward analysis of the first order and second order conditions for this optimization
  problem. Note that $\lambda > 0$ is required for the second order conditions (concavity).
\end{proof}

Denote by $N_T$ the total number of block generated over the time interval $[0,T]$. Suppose an
arbitrageur arrives at time $\tau_i$, observing external price $P_{\tau_i}$ and mispricing
$z_{\tau_i^{-}}$. From \Cref{lem:myopic}, the arbitrageur profit is given by
$A^*(P_{\tau_i},z_{\tau_i^{-}})$ and the trade size is given by $q^*(z_{\tau_i^{-}})$.  We can
write the total arbitrage profit and total quantity traded (measured in the num\'eraire) paid over
$[0,T]$ by summing over all arbitrageurs arriving in that interval, i.e.,
\[
  \ARB_T \defeq \sum_{i=1}^{N_T}
  A^*(P_{\tau_i},z_{\tau_i^{-}}),
  \quad
  \VOL_T \defeq \sum_{i=1}^{N_T}
  P_{\tau_i} q^*(z_{\tau_i^{-}}).
\]
Clearly these are non-negative and monotonically increasing jump processes. The following theorem
characterizes their instantaneous expected rate of growth or intensity:\footnote{Mathematically,
  \bARB is the intensity of the compensator for the monotonically increasing jump process $\ARB_T$
  at time $T=0$, similarly \bVOL is the intensity of the compensator for $\VOL_T$.}
\begin{theorem}[Rate of Arbitrage Profit and Volume]\label{th:arb-rate}
  Define the intensity, or instantaneous rate of arbitrage profit and volume, by
  \[
    \bARB  \defeq \lim_{T\tends 0} \frac{\E\left[ \ARB_T \right]}{T},
    \quad
    \bVOL  \defeq \lim_{T\tends 0} \frac{\E\left[ \VOL_T \right]}{T}.
  \]
  Given initial price $P_0=P$, suppose that $z_{0-}=z$ is distributed according to its stationary
  distribution $\pi(\cdot)$. Then, the instantaneous rate of arbitrage profit and volume are given by
  \begin{equation}\label{eq:arb-rate}
    \bARB
      = \frac{\E_\pi\left[  A^*(P,z)  \right]}{\Delta t}
      = \frac{P r \delta }{\delta - \mu + \tfrac{1}{2} \sigma^2} \times \frac{1}{1 + \zeta_-},
  \end{equation}
  \begin{equation}\label{eq:vol-rate}
    \bVOL
    = \frac{\E_\pi\left[  P q^*(z)  \right]}{\Delta t}
    =   \frac{Pr\delta}{\delta - \mu + \tfrac{1}{2} \sigma^2}.
  \end{equation}
\end{theorem}
Comparing the instantaneous rate of arbitrage profit and volume given by
\eqref{eq:arb-rate}--\eqref{eq:vol-rate} with \Cref{th:regular}, we have that
\begin{equation}\label{eq:vol-lvf}
  \bARB = \bVOL \times \LVF_+.
\end{equation}
This expression highlights the fact a gradual Dutch auction can be viewed as a continuum of many
regular Dutch auctions, each of infinitesimal size, and each at a different price. From
\Cref{th:regular}, we know that the seller will incur the same expected relative loss per dollar
sold, $\LVF_+$, in each of these auctions. This loss is the same irrespective of the different
prices because all of the auctions start out-of-the-money ($z_0 \geq 0$). Equation
\eqref{eq:vol-lvf} intuitively decomposes the total arb profits per unit time as the product of
the dollar volume sold per unit time and the loss per dollar sold.

\paraheader{Parameter optimization.} As in \Cref{sec:rda}, we can leverage \Cref{th:arb-rate} to
optimize parameter choice in a gradual Dutch auction. In particular, a gradual Dutch auction is
parameterized by the choice of emission rate $r \geq 0$ and the choice of drift
$\delta \defeq \lambda + \mu - \tfrac{1}{2} \sigma^2$ satisfying $\delta \geq 0$ and $\lambda =
\delta - \mu + \tfrac{1}{2} \sigma^2 > 0$ (the second condition is required for concavity in \Cref{lem:myopic}).
This choice can be made to minimize the losses incurred while maximizing the rate
of trade. For example, consider the optimization problem
\[
  \minimize_{r > 0,\ \delta \geq \max(0,\mu - \tfrac{1}{2} \sigma^2)}\ \LVF_+ - \theta \cdot \bVOL,
\]
where $\theta \geq 0$ is a tradeoff parameter.

%%% Local Variables:
%%% mode: latex
%%% TeX-master: "pricing-dutch-auctions"
%%% End:

\section{Conclusion and Future Work}

While there has been an increasing amount of academic interest in studying, designing, and formalizing automated market makers for liquid assets in the blockchain context, there has been somewhat less attention paid to Dutch auctions, despite their popularity with protocol implementers. This paper was an attempt to bring the theoretical understanding of Dutch auctions in the setting of discrete block generation times closer to the current level of understanding that has been reached around automated market makers, particularly in \citet{lvr2022} and \citet{milionis2023automated}.

The paper also sought to provide a guide for application designers in setting parameters for Dutch
auctions, including deriving formulas that map the tradeoff between speed of execution and quality
of execution. The paper may also be helpful for platform designers in determining performance parameters like block times. For example, the rule of thumb in \Cref{eq:lvf-bound} suggests that if a platform wants to support Dutch auctions that lose less than 2 basis points for assets with daily volatility of 5\%, it will need to have block times of less than 2.75 seconds.

The model in this paper shared some of the limitations of the model in
\cite{milionis2023automated}, including not taking into account fixed transaction fees such as
``gas'' and use of a Poisson model for block generation as opposed to deterministic block
generation, which is more relevant for modern proof-of-stake blockchains. Further, a purely
diffusive, continuous process (geometric Brownian motion) has been used to model innovations in
the fundamental price process, while jumps are known to be an important component of
high-frequency price dynamics. This paper also assumes that block proposers are independent (and thus short-term profit-maximizing), rather than considering a model in which a single block proposer could acquire control over multiple blocks in a row and use that to extract ``multi-block MEV.''
Finally, while
this work quantified the losses inherent in Dutch auctions, it does not explore possible
alternative designs for Dutch auctions that might mitigate those losses without reducing the speed
of execution. We hope further work can explore this area.

%%% Local Variables:
%%% mode: latex
%%% TeX-master: "pricing-dutch-auctions"
%%% End:

\bibliography{references}

\appendix

\section{Proofs}\label{sec:proofs}

\begin{proof}[\bf\sffamily Proof of \Cref{lem:stationary}]
Note that $\{ z_t \}$ is a Markov jump diffusion process, with infinitesimal generator
\[
\Ascr f(z) = \tfrac{1}{2} \sigma^2 f''(z) - \delta f'(z)
+ \Delta t^{-1} \left[ f(0) - f(z) \right] \I{z < 0},
\]
given a test function $f\colon \R \tends \R$.

% First, we will establish that the process $\{ z_t\}$ is ergodic.  Consider the Lyapunov function
% $V(z) \defeq z^2$. Then,
% \[
%   \Ascr V(z) = \sigma^2 - \delta z - \frac{1}{\Delta t} z^2 \I{z < - 0}
% \]
% i.e., this function satisfies the Foster-Lyapunov negative
% drift condition of Theorem~6.1 of
% \citet{meyn1993stability}. Hence, the process is ergodic and a unique
% stationary distribution exists.
\todo{show stability}

The invariant distribution $\pi(\cdot)$ must satisfy
\begin{equation}\label{eq:stationary}
\E_\pi[\Ascr f(z)] = \int_{-\infty}^{+\infty} \Ascr f(z)\, \pi(dz) = 0,
\end{equation}
for all test functions $f \colon \R \rightarrow \R$.  We will guess that $\pi(\cdot)$ decomposes
according to two different densities over the positive and negative half line, and then compute
the conditional density on each segment via Laplace transforms using \eqref{eq:stationary}.

Consider the test function
\[
f_-(z) \defeq
\begin{cases}
    e^{\alpha z} & \text{if $z < 0$,} \\
    1 + \alpha z & \text{if $z \geq 0$.}
\end{cases}
\]
Then,
\[
\begin{split}
\Ascr f_-(z) & = \tfrac{1}{2} \sigma^2 \alpha^2 e^{\alpha z} \I{z < 0}
- \delta \alpha \left( e^{\alpha z} \I{z < 0} + \I{z \geq 0}   \right)
+ \Delta t^{-1} \left[ 1 - e^{\alpha z} \right] \I{z < 0},
\end{split}
\]
so that
\[
\begin{split}
    0 &  = \E_\pi\left[ \Ascr f_-(z) \right] \\
      & =
        \tfrac{1}{2} \sigma^2 \alpha^2 \pi_- \E_\pi\left[ \left. e^{\alpha z} \right| z < 0
        \right]
        -  \delta \alpha
        \left( \pi_- \E_\pi\left[ \left. e^{\alpha z} \right| z < 0 \right] + \pi_+   \right)
        + \Delta t^{-1} \pi_- \left(  1 - \E_\pi\left[ \left. e^{\alpha z}  \right| z < 0 \right]\right).
\end{split}
\]
Then,
\[
  \begin{split}
    \E_\pi\left[ \left. e^{\alpha z} \right| z < 0 \right]
    & =
      \frac{\delta \alpha \frac{\pi_+}{\pi_-} - \Delta t^{-1} }
      { \tfrac{1}{2} \sigma^2 \alpha^2  - \delta \alpha  - \Delta t^{-1}  }.
  \end{split}
\]
Observe the denominator has a single negative root. Then, $\pi(-z|z<0)$ must be exponential with
parameter
$\zeta_-\defeq \left( \sqrt{\delta^2 + 2 \Delta t^{-1} \sigma^2} - \delta \right) /
\sigma^2$. Also, note that
\[
\E_\pi[-z|z<0] = 1/\zeta_-.
\]

Next consider the test function
\[
f_+(z) \defeq
\begin{cases}
    e^{-\alpha z} & \text{if $z \geq 0$,} \\
    1 - \alpha z & \text{if $z < 0$.}
\end{cases}
\]
Then,
\[
\begin{split}
\Ascr f_+(z) & = \tfrac{1}{2} \sigma^2 \alpha^2 e^{-\alpha z} \I{z \geq 0}
+ \delta \alpha \left( e^{-\alpha z} \I{z \geq 0} + \I{z < 0}   \right)
+ \Delta t^{-1} \alpha z \I{z < 0},
\end{split}
\]
so that
\[
\begin{split}
    0 &  = \E_\pi\left[ \Ascr f_+(z) \right] \\
      & =
        \tfrac{1}{2} \sigma^2 \alpha^2 \pi_+ \E_\pi\left[ \left. e^{-\alpha z} \right| z \geq 0
        \right]
        +  \delta \alpha
        \left( \pi_+ \E_\pi\left[ \left. e^{-\alpha z} \right| z \geq 0 \right] + \pi_-   \right)
        + \Delta t^{-1} \alpha \pi_- \E_\pi\left[ \left. z  \right| z < 0 \right].
\end{split}
\]
Then,
\[
  \begin{split}
  \E_\pi\left[ \left. e^{-\alpha z} \right| z \geq 0 \right]
    & = -\frac{\pi_-}{\pi_+} \frac{\delta
     + \Delta t^{-1} \E_\pi\left[ \left. z  \right| z < 0 \right]}
      {\tfrac{1}{2} \sigma^2 \alpha   + \delta} \\
    & = -\frac{\pi_-}{\pi_+} \frac{\delta
     - \Delta t^{-1}/\zeta_-}
      {\tfrac{1}{2} \sigma^2 \alpha   + \delta} \\
  \end{split}
\]
Then, $\pi(z|z\geq 0)$ must be exponential with parameter $\zeta_+\defeq
2\delta/\sigma^2$. Substituting $\alpha = 0$, we have
\[
1 = -\frac{\pi_-}{\pi_+} \frac{\delta
     - \Delta t^{-1}/\zeta_-}
   {\delta}
   =
   -\frac{\pi_-}{1-\pi_-}
   \left(
     1
     -
   \frac{1}{\delta \Delta t \zeta_-}
    \right).
\]
Solving for $\pi_-$,
\[
  \pi_- = \delta \Delta t \zeta_-,
  \quad
  \pi_+ = 1 - \delta \Delta t \zeta_-.
\]

\end{proof}

\begin{proof}[\bf\sffamily Proof of \Cref{th:regular}]
  We consider $\LVF(z_0)$ and $\FT(z_0)$ separately.

  \paraheader{Loss-versus-fair.}  We begin with the \LVF calculation. First, consider the case
  where $z_0 \geq 0$. Define the $\tau_F$ to be the fill time of the order, i.e., the first
  Poisson block generation time $\tau_F$ with $z_{\tau_F} \leq 0$. Also define
  $\tau_0 \defeq \min\{ t \geq 0 \colon z_t = 0 \}$ to be first passage time for the boundary
  $z_t=0$. Since the the process the mispricing process is continuous, we must have that
  $\tau_F \geq \tau_0$. Then,
  \begin{equation}\label{eq:split-lvf-p}
    \begin{split}
      \LVF(z_0)
      & \defeq \E\left[ \left. 1 - e^{z_{\tau_F}} \right| z_0\right] \\
      & \stackrel{\textrm{(a)}}{=}
        \E\left[ \left. \E\left[ \left. 1 - e^{z_{\tau_F}} \right| \tau_0,z_{\tau_0} \right] \right|
        z_0\right] \\
      & \stackrel{\textrm{(b)}}{=}  \E\left[ \left. \LVF(z_{\tau_0}) \right| z_0\right] \\
      & \stackrel{\textrm{(c)}}{=} \LVF(0) \defeq \LVF_+.
    \end{split}
  \end{equation}
  where (a) follows from the tower property of expectation,  (b)
  follows from the fact that Poisson arrivals are memoryless and $\{ z_t\}$ is a Markov process,
  and (c) follows from the fact that $z_{\tau_0} = 0$.

  Now, consider arbitrary $z_0\in\R$. Let $\tau_B > 0$ be the first Poisson block generation
  time. Since $\tau_F \geq \tau_B$,
  \begin{equation}\label{eq:split-lvf-p-2}
    \begin{split}
      \LVF(z_0)
      & \defeq \E\left[ \left. 1 - e^{z_{\tau_F}} \right| z_0\right] \\
      & \stackrel{\textrm{(a)}}{=} \E\left[ \left.
        \E\left[ \left. 1 - e^{z_{\tau_F}}
        \right| \tau_B,z_{\tau_B}
        \right]
        \right| z_0\right] \\
      & \stackrel{\textrm{(b)}}{=}  \E\left[ \left.
        \LVF(z_{\tau_B})
        \right| z_0\right] \\
      & \stackrel{\textrm{(c)}}{=} \E\left[ \left.
        \LVF_+ \I{z_{\tau_B} \geq 0}
        +
        \left( 1 - e^{z_{\tau_B}} \right) \I{z_{\tau_B} < 0}
        \right| z_0\right] \\
      & \stackrel{\textrm{(d)}}{=}
        \int_{0}^\infty \frac{e^{-\tau/\Delta t}}{\Delta t}
        \int_0^{+\infty} \frac{1}{\sigma \sqrt{2\pi \tau}}
        e^{-\frac{1}{2} \left( \frac{z + \delta \tau - z_0}{\sigma \sqrt{\tau}}\right)^2}
        \LVF_+
        \, dz\, d\tau
      \\
      & \quad
        +
        \int_{0}^\infty \frac{e^{-\tau/\Delta t}}{\Delta t}
        \int_{-\infty}^0 \frac{1}{\sigma \sqrt{2\pi \tau}}
        e^{-\frac{1}{2} \left( \frac{z + \delta \tau - z_0}{\sigma \sqrt{\tau}}\right)^2}
        \left( 1 - e^{z} \right)
        \, dz\, d\tau.
      \\
      % & =
      %   \int_{0}^\infty \frac{e^{-\tau/\Delta t}}{\Delta t}
      %   \left\{
      %   \LVF_+
      %   \left(
      %   1
      %   -
      %    \Phi\left( \frac{\delta\tau - z_0}{\sigma\sqrt{\tau}} \right)
      %   \right)
      %   +
      %   \Phi\left( \frac{\delta\tau - z_0}{\sigma\sqrt{\tau}} \right)
      %   \right\}
      %   \,
      %   d\tau
      % \\
      % & \quad
      %   +
      %   \int_{0}^\infty \frac{e^{-\tau/\Delta t}}{\Delta t}
      %   e^{\delta \tau - z_0 + \tfrac{1}{2} \sigma^2 \tau}
      %   \, d\tau
      % \\
      % & =
      %   \int_{0}^\infty \frac{e^{-\tau/\Delta t}}{\Delta t}
      %   \left\{
      %   \LVF_+
      %   \left(
      %   1
      %   -
      %    \Phi\left( \frac{\delta\tau - z_0}{\sigma\sqrt{\tau}} \right)
      %   \right)
      %   +
      %   \Phi\left( \frac{\delta\tau - z_0}{\sigma\sqrt{\tau}} \right)
      %   \right\}
      %   \,
      %   d\tau
      % \\
      % & \quad
      %   +
      %   \frac{e^{-z_0}}{1 - \Delta t \left( \delta - \tfrac{1}{2} \sigma^2\right)}
    \end{split}
  \end{equation}
  where (a) follows from the tower property of expectation,  (b)
  follows from the fact that Poisson arrivals are memoryless and $\{ z_t\}$ is a Markov process,
  (c) follows from the fact that  $\LVF(z_{\tau_B})=\LVF_+$ for $z_{\tau_B} \geq 0$ while
  $\LVF(z_{\tau_B})=1-e^{z_{\tau_B}}$ if  $z_{\tau_B} < 0$, (d) follows from the fact that
  $\tau_B$ is exponentially distributed while, conditional on $\tau_B$, $z_{\tau_B}$ is normally
  distributed, and $\Phi(\cdot)$ is the cumulative normal distribution. Substituting in $z_0=0$
  and solving for $\LVF(z_0) = \LVF_+$, after integration, we obtain \eqref{eq:lvf-p}. For $z_0
  \leq 0$, we can substitute \eqref{eq:lvf-p} into \eqref{eq:split-lvf-p-2} and integrate to
  obtain \eqref{eq:lvf-n}.

  % \LVF_+ under stationary distribution: https://www.desmos.com/calculator/64gsttchtf
  % negative case: https://www.desmos.com/calculator/zbu4ww8zcf

\paraheader{Time-to-fill.}
Suppose we start out at $z_0 \geq 0$, and define $\FT(z_0)$ to be the expected fill time of the next
trade, i.e., the first Poisson arrival time $\tau$ with $z_{\tau_F} \leq 0$. Also define $\tau_0 = \min\{ t
  \geq 0 \colon z_t = 0 \}$ to be the first passage time for the boundary $z_t=0$. Then, since the mispricing process $\{z_t\}$
is continuous and Markov, and Poisson arrivals are memoryless, we have that $\tau_F \geq \tau_0$ and
\[
  \FT(z_0)
  = \E\left[ \tau_F | z_0 \right]
  = \E\left[ \tau_0 | z_0 \right]
  + \E\left[ \left. \E\left[ \left. \tau_F - \tau_0 \right| \tau_0,z_{\tau_0} \right]
      \right| z_0 \right]
  %= \E\left[ \tau_0 | z_0 \right] + \FT(0)
  = \frac{z_0}{\delta} + \FT(0),
\]
where we have used the standard formula for expected first passage time of a Brownian motion with
drift.

Thus, we have reduced to the case where $z_0=0$. Define $\tau_B>0$ to be first Poisson block
generation time. If $z_{\tau_B} \leq 0$, then $\tau_B$ is also the fill time. On the other hand, if
$z_{\tau_B} > 0$, we will have to wait an additional amount after $\tau_B$ given in expectation by
$\FT(z_{\tau_B})=\FT(0)$. Thus, integrating over $\tau_B$, and using the fact that, given $\tau_B$,
$z_{\tau_B}$ is normally distributed,
\[
  \begin{split}
  \FT(0) & = \int_0^\infty \frac{e^{-\tau/\Delta t}}{\Delta t}
  \left( \tau
           + \int_0^\infty \frac{1}{\sigma \sqrt{2\pi \tau}}
           e^{-\frac{1}{2} \left( \frac{z + \delta \tau}{\sigma \sqrt{\tau}}\right)^2}
           \FT(z)\, dz \right)\, d\tau \\
    & = \Delta t +
      \int_0^\infty \frac{e^{-\tau/\Delta t}}{\Delta t}
      \int_0^\infty \frac{1}{\sigma \sqrt{2\pi \tau}}
      e^{-\frac{1}{2} \left( \frac{z + \delta \tau}{\sigma \sqrt{\tau}}\right)^2}
      \left(  \frac{z}{\delta} + \FT(0) \right)\, dz \, d\tau \\
    & = \Delta t +
      \int_0^\infty \frac{e^{-\tau/\Delta t}}{\Delta t}
      \left\{
      \int_0^\infty \frac{1}{\sigma \sqrt{2\pi \tau}}
      e^{-\frac{1}{2} \left( \frac{z + \delta \tau}{\sigma \sqrt{\tau}}\right)^2}
      \frac{z}{\delta} \, dz
      +
      \FT(0)
      \left( 1 - \Phi\left( \frac{\delta \sqrt{\tau}}{\sigma} \right) \right)
      \right\}\, d\tau \\
    & = \Delta t +
      \int_0^\infty \frac{e^{-\tau/\Delta t}}{\Delta t}
      \left\{
      \int_0^\infty \frac{1}{\sigma \sqrt{2\pi \tau}}
      e^{-\frac{1}{2} \left( \frac{z + \delta \tau}{\sigma \sqrt{\tau}}\right)^2}
      \frac{z}{\delta} \, dz
      +
      \FT(0)\,
      \Phi\left( -\frac{\delta \sqrt{\tau}}{\sigma} \right) \right\}
      \, d\tau. \\
  \end{split}
\]
We can solve this for $\FT(0)$, i.e.,
\[
  \FT(0) = \frac{\displaystyle
    \Delta t +
    \int_0^\infty \frac{e^{-\tau/\Delta t}}{\Delta t}
      \int_0^\infty \frac{1}{\sigma \sqrt{2\pi \tau}}
      e^{-\frac{1}{2} \left( \frac{z + \delta \tau}{\sigma \sqrt{\tau}}\right)^2}
      \frac{z}{\delta} \, dz \, d\tau
    }{\displaystyle
      1
      -
      \int_0^\infty \frac{e^{-\tau/\Delta t}}{\Delta t}
      \Phi\left( -\frac{\delta \sqrt{\tau}}{\sigma} \right)
    \, d\tau
  }
  =
    \frac{\Delta t}{2} \left( 1 + \sqrt{1 + \frac{2 \sigma^2}{\delta^2 \Delta t}} \right),
\]
where the final equality is obtained via integration. This establishes \eqref{eq:ft-p}.

Finally, consider the case where $z_0 < 0$.
Define $\tau_B>0$ to be first Poisson block
generation time. If $z_{\tau_B} \leq 0$, then $\tau_B$  is also the fill time. On the other hand, if
$z_{\tau_B} > 0$, we will have to wait an additional amount after $\tau_B$ given in expectation by
$\FT(z_{\tau_B})$. Thus,
\[
  \begin{split}
    \FT(z_0)
    & = \int_0^\infty \frac{e^{-\tau/\Delta t}}{\Delta t}
      \left( \tau
      + \int_0^\infty \frac{1}{\sigma \sqrt{2\pi \tau}}
      e^{-\frac{1}{2} \left( \frac{z + \delta \tau - z_0}{\sigma \sqrt{\tau}}\right)^2}
      \FT(z)\, dz \right)\, d\tau \\
    & = \Delta t +
      \int_0^\infty \frac{e^{-\tau/\Delta t}}{\Delta t}
      \int_0^\infty \frac{1}{\sigma \sqrt{2\pi \tau}}
      e^{-\frac{1}{2} \left( \frac{z + \delta \tau - z_0}{\sigma \sqrt{\tau}}\right)^2}
      \left(  \frac{z}{\delta} + \FT(0) \right)\, dz \, d\tau \\
    & = \Delta t +
      \int_0^\infty \frac{e^{-\tau/\Delta t}}{\Delta t}
      \left\{
      \int_0^\infty \frac{1}{\sigma \sqrt{2\pi \tau}}
      e^{-\frac{1}{2} \left( \frac{z + \delta \tau - z_0}{\sigma \sqrt{\tau}}\right)^2}
      \frac{z}{\delta} \, dz
      +
             \FT(0)\,
                     \Phi\left( - \frac{\delta \tau - z_0}{\sigma \sqrt{\tau}} \right)
              \right\}
      \, d\tau.
  \end{split}
\]
After integration, this yields \eqref{eq:ft-n}.
\end{proof}

\begin{proof}[\bf\sffamily Proof of \Cref{th:arb-rate}]
  We follow the method of \citet{milionis2023automated}.
  %In particular, we can express
  %$\ARB_T$ as an integral against the increments of the Poisson arrival counting process.
  Specifically, using the smoothing
  formula, e.g., Theorem~13.5.7 of \citet{bremaud2020markov},
  \[
    \E\left[ \ARB_T \right]
    =
    \E\left[
      \sum_{i=1}^{N_T}
      A^*(P_{\tau_i},z_{\tau_i-})
    \right]
    =
    \E\left[
      \int_0^T
      A^*(P_{t},z_{t-})
      \, dN_t
    \right]
    =
    \E\left[
      \int_0^T
      A^*(P_{t},z_{t-})
      \cdot
      \Delta t^{-1}
      \, dt
    \right].
  \]
  Applying Tonelli's theorem and the fundamental theorem of calculus,
  \[
    \bARB \defeq
    \lim_{T\tends 0} \frac{\E\left[ \ARB_T \right]}{T}
    = \lim_{T\tends 0} \frac{1}{T}
    \int_0^T
    \E\left[
      \frac{A^*(P_{t},z_{t-})}{\Delta t}
    \right]
    \, dt
    =
    \frac{
    \E\left[
      A^*(P_{0},z_{0-})
    \right]
    }{\Delta t}
    =
    \frac{
    \E_\pi\left[
      A^*(P,z)
    \right]
    }{\Delta t}
    ,
  \]
  where in the final expression, $P_0=P$ and $z$ is distributed according to the stationary
  distribution $\pi(\cdot)$.

  Then, using the definition of $\pi(\cdot)$ from \Cref{lem:stationary} and $A^*(\cdot,\cdot)$ from
  \Cref{lem:myopic},
  % reference: https://www.desmos.com/calculator/xdlhwfm2sh
  \[
    \begin{split}
      \bARB
      &
        = \frac{1}{\Delta t} \pi(z | z \leq 0)
         \frac{Pr}{\lambda}
        \E_\pi\left[ \left. e^z - 1 - z \right| z \leq
        0\right]
      \\
      % &
      %   = \delta \zeta_{-} \frac{Pr}{\lambda}
      %   \int_0^{+\infty}  \zeta_{-} e^{-\zeta_{-}z} \left\{ e^{-z} - 1 + z \right\} \, dz
      % \\
      & =  \frac{Pr}{\lambda}
        \delta \zeta_{-}
        \left\{
        \frac{\zeta_{-}}{1 + \zeta_{-}} - 1 + \frac{1}{\zeta_{-}}
        \right\}
      \\
      & = \frac{Pr}{\lambda} \times \frac{\delta}{1 + \zeta_{-}},
    \end{split}
  \]
  as desired.

  The same argument establishes that
  \[
    \bVOL\defeq \lim_{T\tends 0} \frac{\E\left[ \VOL_T \right]}{T}
    = \frac{\E_\pi\left[  P q^*(z)  \right]}{\Delta t}.
  \]
  Then, using \Cref{lem:stationary} and $q^*(\cdot)$ from
  \Cref{lem:myopic},
  \[
    \begin{split}
      \bVOL
      & = \frac{1}{\Delta t} \pi(z | z \leq 0)
        \frac{Pr}{\lambda}
        \E_\pi\left[ \left. - z \right| z \leq 0\right]
      \\
      & = \frac{Pr}{\lambda} \times \delta,
    \end{split}
  \]
  as desired.
\end{proof}

\section{Formulas Under Fundamental Value Uncertainty}\label{sec:parameter}

Assume that the prior belief on the initial mispricing is normally distributed, i.e.,
$z_0\sim N(\mu_0,\sigma_0^2)$. Then, via direct integration,
% reference: https://www.desmos.com/calculator/bw17oc46pn
\[
  \begin{split}
    \E\left[ \LVF\left(z_0\right) \right]
    & =
      \LVF_++\left(1-\LVF_+\right) \Phi\left(-\frac{\mu_{0}}{\sigma_0}\right)
      +\frac{ \Delta t^{-1} e^{\frac{\sigma_{0}^{2}}{2}}}{\frac{\sigma^{2}}{2}- \delta - \Delta t^{-1} } e^{\mu_{0}}
      \Phi\left(-\sigma_{0}-\frac{\mu_{0}}{\sigma_{0}}\right)
    \\
    & \quad
      +
      \Biggl\{
      \left(\LVF_+-\frac{\frac{\sigma^{2}}{2}- \delta }{\frac{\sigma^{2}}{2}- \delta - \Delta
      t^{-1} }\right)e^{\frac{ \Delta t^{-1} \sigma_{0}^{2}}{\sigma^{2}}+\frac{ \delta
      }{\sigma^{2}}\left(\frac{ \delta
      }{\sigma^{2}}\sigma_{0}^{2}+\mu_{0}\right)\left(\sqrt{1+\frac{2 \Delta t^{-1} \sigma^{2}}{
      \delta ^{2}}}+1\right)}
    \\
    & \qquad\qquad
      \times \Phi\left(-\frac{ \delta \sigma_{0}}{\sigma^{2}}\left(\sqrt{1+\frac{2 \Delta t^{-1} \sigma^{2}}{ \delta ^{2}}}+1\right)-\frac{\mu_{0}}{\sigma_{0}}\right)      \Biggr\}.
  \end{split}
\]
Similarly,
% reference: https://www.desmos.com/calculator/pqw6nk0nf4
\[
  \begin{split}
    \E\left[ \FT\left(z_0\right) \right]
    & =
       \Delta t+\frac{\sigma_{0}}{ \delta \sqrt{2\pi}}e^{-\frac{\mu_{0}^{2}}{2\sigma_{0}^{2}}}
    \\
    & \quad
      +\left(\FT(0)- \Delta t+\frac{\mu_{0}}{ \delta
      }\right)\Phi\left(\frac{\mu_{0}}{\sigma_{0}}\right)
    \\
    & \quad
      +\left(\FT(0)- \Delta t\right)e^{2 \Delta t^{-1} \FT(0)\left(\frac{\sigma_{0}^{2} \delta ^{2}}{\sigma^{4}}+\frac{\mu_{0} \delta }{\sigma^{2}}\right)+\frac{ \Delta t^{-1} \sigma_{0}^{2}}{\sigma^{2}}}\Phi\left(-\frac{2 \Delta t^{-1} \FT(0)\sigma_{0} \delta }{\sigma^{2}}-\frac{\mu_{0}}{\sigma_{0}}\right).
  \end{split}
\]

\end{document}